\documentclass[aps,pre,reprint,a4paper,showpacs,notitlepage,superscriptaddress]{revtex4-1}
\usepackage{amsmath,amssymb}
\usepackage[dvipdfmx]{graphicx}
\usepackage{bm}
\usepackage{theorem}
\newtheorem{theorem}{Theorem}
\makeatletter
\newenvironment{proof}[1][\proofname]{\par
  \normalfont
  \topsep6\p@\@plus6\p@ \trivlist
  \item[\hskip\labelsep{\bfseries #1}\@addpunct{.}]\ignorespaces
}{%
  \endtrivlist
}
\newcommand{\proofname}{\it{Proof}}
\makeatother
\def\qed{\hfill{$\square$}}

\usepackage{color}	

\begin{document}

\title{Typical Approximation Performance for Maximum Coverage Problem}
\date{\today}

\author{Satoshi Takabe}
\email{E-mail: s{\_}takabe@nitech.ac.jp}
\affiliation{Department of Computer Science, Nagoya Institute of Technology, Gokiso-cho, Showa-ku, Nagoya, Aichi, 466-8555, Japan}
\affiliation{RIKEN Center for Advanced Intelligence Project 15F, 1-4-1, Nihonbashi, Chuo-ku, Tokyo, 103-0027, Japan
}
\author{Takanori Maehara}
\affiliation{RIKEN Center for Advanced Intelligence Project 15F, 1-4-1, Nihonbashi, Chuo-ku, Tokyo, 103-0027, Japan
}
\author{Koji Hukushima}
\affiliation{Graduate School of Arts and Sciences, The University of Tokyo, 3-8-1 Komaba, Meguro-ku, Tokyo 153-8902, Japan}

\begin{abstract}
This study investigated typical performance of approximation algorithms known as belief propagation, greedy algorithm, and
linear-programming relaxation for maximum coverage problems on sparse biregular random graphs.
After using the cavity method for a corresponding hard-core lattice--gas model, 
results show that two distinct thresholds of replica-symmetry and its breaking exist
 in the typical performance threshold of belief propagation.
 In the low-density region, the superiority of three algorithms
 in terms of a typical performance threshold is obtained by some theoretical analyses.
 Although the greedy algorithm and linear-programming relaxation have the
 same approximation ratio in worst-case performance,
  their typical performance thresholds are mutually different, indicating the importance of typical performance.
Results of numerical simulations validate the theoretical analyses and imply further mutual relations of approximation algorithms 
\end{abstract}


\maketitle
\section{Introduction}
Approximation algorithms, which are important for hard optimization problems, have attracted researchers' interest
because, for NP-hard optimization problems, one 
encounters difficulty when solving optimal solutions exactly in polynomial time of the problem size.
Since the P versus NP problem arose, development and performance analyses of approximation algorithms
have persisted as a central issue of computer science and operations research.

Two performance evaluations of approximation algorithms exist: worst-case performance and typical (or average-case) performance.
As described in this paper, we specifically examine the latter mainly using statistical--mechanical methods, 
 although the former has been investigated mainly in the literature of theoretical computer science~\cite{Vazirani2003}.
Worst-case performance is defined by the pair of an optimization problem and its approximation algorithm.
An approximation ratio is then defined by the maximal ratio of an optimal value to an approximation value
over all instances if the problem is a maximization problem (and vice versa, otherwise).
It is important to provide strict performance guarantee of the algorithm,
 although the worst-case instance is sometimes pathological.
However, the typical performance is defined as the average performance of approximation algorithm
 for a given optimization problem over randomized instances.
It is sometimes useful for practical use. It sheds light on properties of optimization problems and
 approximation algorithms in a perspective that is different from worst-case analysis.

 An interesting point of the typical property of optimization problems
 is found in a close relation to the spin-glass theory in statistical physics.
Extensive studies of various computational problems have revealed that the concept of replica symmetry (RS) and
its breaking (RSB) in spin-glass theory reflects average computational complexity~\cite{Monasson1999} and 
structure of the solution space~\cite{Mezard2003} of the problems.

Typical performance of approximation algorithms often exhibits a phase transition on typical goodness
 or accuracy of approximation, which is called threshold phenomena in the literature of theoretical computer science~\cite{Karp1981}.
Some approximation algorithms for minimum vertex covers (min-VC) are good examples of threshold phenomena.
The problem is defined on a graph. The randomized problem is
 characterized using a random graph ensemble with $c$ being the average degree.
From a statistical--mechanical perspective, the problem on Erd\"os-R\'enyi random graphs
 has the RS-RSB threshold at $c=e=2.71\dots$~\cite{Weigt2000a}.
Moreover, three approximation algorithms have been investigated: belief propagation (BP)
 (or message passing)~\cite{Weigt2006}, greedy leaf-removal algorithm~\cite{Bauer2001}, and
 linear-programming (LP) relaxation~\cite{Dewenter2012,Takabe2014a}.
Approximation algorithms other than BP naively have no direct connection to the spin-glass theory.
Nevertheless, for Erd\"os-R\'enyi random graphs, these algorithms have the same performance threshold
as the RS-RSB threshold.
However, subsequent studies of general random graphs~\cite{Takabe2016a}
indicate that their thresholds are not equivalent for some random graphs because of their graph structure.
Using statistical--mechanical techniques, typical performance has been studied
 for variants of leaf removal~\cite{Liu2012,Takabe2014,Lucibello2014,Habibulla2015,Takahashi}
 and relaxation technique such as LP relaxation~\cite{Schawe2016,Takabe2016} and semidefinite-programming relaxation~\cite{Javanmard2015}.
These studies have revealed not only typical approximate performance itself but also a suggestive connection
 to typical properties of optimization problems, random graph structure, and the spin-glass theory.

As described in this paper, we examine the unweighted maximum coverage (max-COV) problem defined in the next section.
Although the max-COV belongs to the class of NP-hard, 
it has several practical applications such as pan and scan problems~\cite{Johnson2011a}, 
multi-topic blog watch~\cite{Saha2009}, and text summarization~\cite{Takamura2009}.
Our main purpose is to examine the typical performance of approximation algorithms for the problem and to compare their performance thresholds.
We specifically examine three approximation algorithms:
belief propagation, greedy algorithm, and LP relaxation.
Some statistical--mechanical methods applied to both analytical and
numerical analyses together with mathematical rigorous discussions clarify typical performance of those algorithms 
and a suggestive mutual relation among approximation algorithms.

This paper is organized as follows. 
In the next section, we define details of the max-COV and its approximation algorithms.
As a random graph ensemble, biregular random graphs are also defined.
In Section~\ref{sec_bp}, a hard-core lattice--gas model for the problem is introduced. Its BP equations are obtained 
based on Bethe--Peierls approximation.
In Section~\ref{sec_rs}, we present a study of the typical performance of BP using the RS cavity method.
Calculation of the spin-glass susceptibility provides the threshold below which BP typically approximate max-COV with high accuracy.
In Section~\ref{sec_gre}, the greedy algorithm is analyzed based on a mean-field rate equation of its deletion process.
Using the obtained solution, the typical performance threshold of the greedy algorithm is evaluated.
In Section~\ref{sec_lp}, LP-relaxed approximate values are evaluated rigorously. The theorem is proved using the weak duality theorem.
In Section~\ref{sec_num}, we describe some numerical results which 
support the validity of theoretical analyses presented in the previous sections.
We also execute some additional simulations to consider the typical performance of a modified greedy algorithm and randomized rounding of LP relaxation.
These results provide suggestive relation between approximation algorithms. 
The last section is devoted to a summary and discussion of the results.

 \section{
 Max-Cover problem and 
 approximation algorithms}
\subsection{Maximum coverage problem}
The max-COV is defined as follows: 
Let $S$ be a set of $M$ elements and $\mathcal{S}=\{S_1,\cdots,S_N\}$ be a collection of subsets of $S$.
For a given positive integer $K(\le N)$, the problem is to choose at most $K$ subsets to maximize the total number of elements
 in the union of chosen subsets.
An example of an instance with $M=4$ and $N=3$ is shown at the left-hand side of Fig.~\ref{fig_0}.
Given that $K=2$, one can select subsets $S_1$ and $S_3$ such that all elements are included in the union of the subsets.

\begin{figure}
 \includegraphics[width=0.8\linewidth]{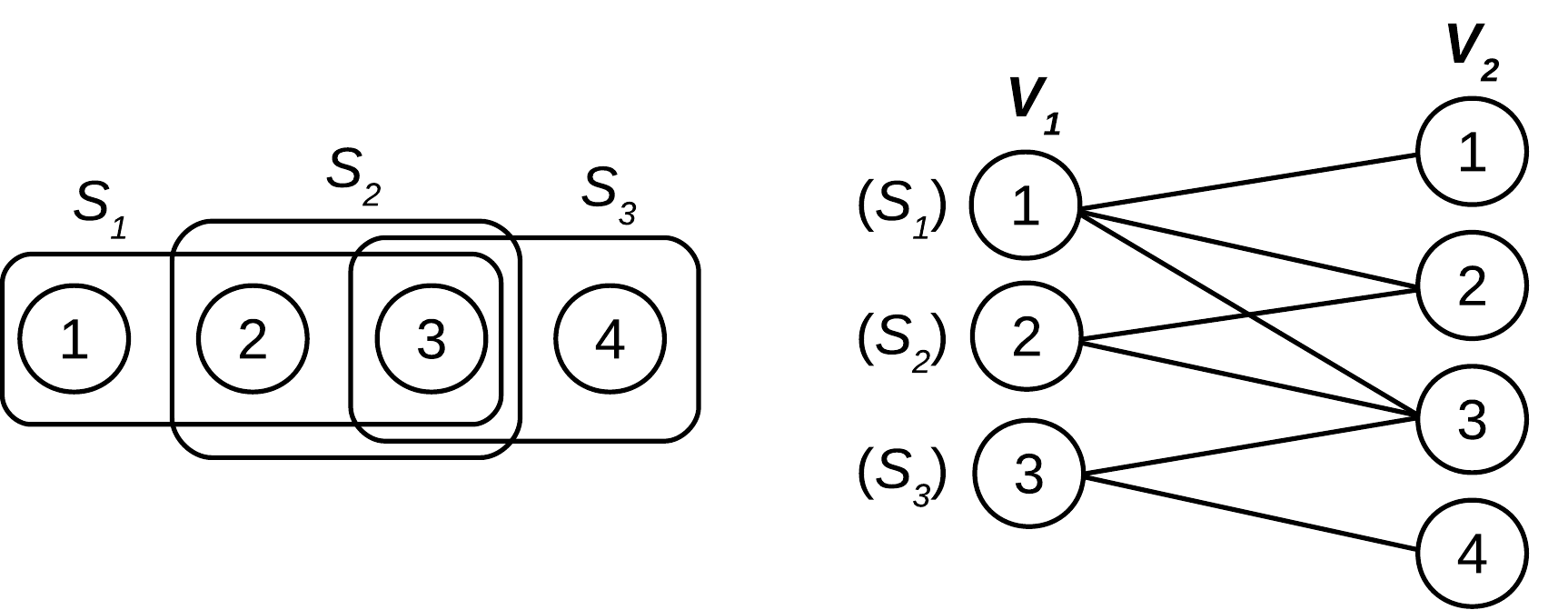}
\caption{Example of the maximum coverage problem. (Left) An example with $M=4$ and $N=3$. $S_1$, $S_2$, and $S_3$ represent subsets of $S=\{1,2,3,4\}$.
(Right) A bipartite graph corresponding to the left example. Each vertex in $V_1$ stands for a subset in the left example.}\label{fig_0}
\end{figure}

Converting an instance to a bipartite graph, one obtains an integer-programming (IP) representation of the max-COV.
Let $G=(V_1,V_2,E)$ be a bipartite graph where each edge in $E$ connects to a vertex in $V_1$ and a vertex in $V_2$.
Vertices in $V_1$ and $V_2$ are labeled respectively by $i\in\{1,\dots,N\}$ and $a\in\{1,\dots,M\}$.
$V_1$ corresponds to a collection of subsets $\mathcal{S}$, whereas $V_2$ represents $S$.
Each edge $(i,a)$ is set if subset $S_i$ includes element $a$.
For Fig.~\ref{fig_0}, the left example is converted to a bipartite graph in the right-hand side.
We then introduce binary variables $\{x_i\}$ and $\{y_a\}$, respectively, to $V_1$ and $V_2$.
Also, $x_i$ is set to one if vertex $i$ (or subset $S_i$) is selected and zero otherwise.
Similarly, $y_a$ is set to one if vertex $a$ is connected to a selected vertex (or element $a$ belongs to
the union of chosen subsets) and zero otherwise.
The problem is therefore represented by the following IP problem:
\begin{equation}
\begin{array}{lll}
\mbox{Max.} &\displaystyle \sum_{a=1}^My_a, &\\
\mbox{subject to\ } &\displaystyle y_a\le \sum_{i;\,(i,a)\in E}x_i & \forall a\in V_2,\\
&\displaystyle \sum_{i=1}^N x_i\le K,&\\
&x_i\in\{0,1\}\:\forall i\in V_1,& y_a\in\{0,1\}\:\forall a\in V_2.
\end{array}
\label{eq_s1_1}
\end{equation}
The inequality in the second constraint can be replaced with an equality
sign because the problem is a maximization problem.

As described in this paper, typical-case property of the max-COV is analyzed by randomizing its instance.
Then, we introduce $\rho_x =K/N$ and assume that $\rho_x\in[0,1]$ is constant.
For random bipartite graphs, we specifically examine $(L,R)$-biregular
random graphs where degrees of each vertex in $V_1$ and $V_2 $ respectively denote
$L$ and $R$.
Using this simple random graph ensemble, our statistical--mechanical analyses
reveal interesting properties of the problem.

To take the large graph limit, the number of vertices in $V_2$ is
rescaled as $M=\alpha N$
with a constant factor $\alpha$. 
We assume that a random graph is sparse, i.e., $L$ and $R$ are constant with respect to $N$.
For the randomized max-COV, we define an average optimal cover ratio
over random graphs with cardinality $N+M$ and a given $\rho_x$ by
\begin{equation}
\rho_y(\rho_x;N)=\frac{1}{M}\overline{\sum_{a=1}^My_a^{\mathrm{opt}}(G,\rho_x)}, \label{eq_s1_2}
\end{equation}
where $\{y_a^{\mathrm{opt}}(G,\rho_x)\}$ represents optimal solutions in $V_2$ and $\overline{(\cdots)}$ represents
 an average over random bipartite graphs with size $N+M$.
Its limiting value to $N\rightarrow\infty$ is denoted by $\rho_y$. We
 simply call it the average optimal cover ratio.

\subsection{Approximation algorithms}\label{sec_app}
We introduce three approximation methods for max-COVs.
The first algorithm is belief propagation (BP).
The recursive equations called BP equations are derived from Bethe--Peierls approximation
 for a spin system corresponding to a given optimization problem.
For systems on trees, graphs with no cycles, the Bethe--Peierls approximation and BP are exact.
In general, however, there exist cycles in a graph yielding
correlations between variables or spins, which results in inexact estimation of solutions (or configurations).
BP on graphs with some cycles is regarded as an approximation method. It is called loopy BP~\cite{Mezard2009}.

The second one is a simple greedy algorithm. 
At each step, this algorithm has the following procedure: (i) choose one vertex named $i$ with the maximum degree in $V_1$,
(ii) delete vertices neighboring to vertex $i$ from $V_2$,
and (iii) update $V_1$ to $V_1\backslash i$ and return to (i) if $|V_1|>N-K$.
This simple algorithm gives an approximation ratio of $1-1/e$~\cite{Nemhauser1978}.
The problem cannot be approximated within this ratio unless P$=$NP~\cite{Feige1998}.

The last is linear programming (LP) relaxation.
An integer programming problem including the max-COV is relaxed to LP problems by replacing integral constraints with real constraints.
The LP-relaxed max-COV is given as
\begin{equation}
\begin{array}{lll}
\mbox{Max.} &\displaystyle \sum_{a=1}^My_a, &\\
\mbox{subject to\ } &\displaystyle y_a\le \sum_{i;\,(i,a)\in E}x_i & \forall a\in V_2,\\
&\displaystyle \sum_{i=1}^N x_i\le K,&\\
&x_i\in [0,1]\:\forall i\in V_1,& y_a\in[0,1]\:\forall a\in V_2.
\end{array}
\label{eq_s1_3}
\end{equation}
Actually, LP approximation value gives the upper bound of the problem.
Because LP problems are solvable exactly in polynomial time, LP relaxation is a widely used approximation technique.
The approximate solution obtained by LP relaxation usually involves non-integers.
One must round those numbers appropriately to obtain an approximate integral solution for the IP problem.
Here, we consider randomized rounding~\cite{Raghavan1987}.
Using the obtained approximation solution $\{x_i^{\mathrm{LP}}\}$, one selects vertex $i\in V_1$ to set $I\subset V_1$
 with probability $x_i^{\mathrm{LP}}/K$ up to $K$ vertices.
Then, the rounded solution $\{x_i^{\mathrm{LPr}}\}$ is set to $x_i^{\mathrm{LPr}}=1$ for $i\in I$
 and $x_i^{\mathrm{LPr}}=0$ otherwise.
 The rounded approximation value is readily calculated from $\{x_i^{\mathrm{LPr}}\}$.
In terms of the worst-case performance, LP relaxation and its randomized rounding in expectation has the same approximation ratio as the greedy algorithm.

As described in this paper, we study the typical performance of these approximation algorithms.
It is evaluated by approximate values averaged over randomized max-COVs defined in the last subsection.
Similar to the average optimal cover ratio, the average cover ratio is defined as the average ratio of the approximate value
 to the cardinality $M$ of vertex set $V_2$.
It is regarded as exhibiting good typical performance if the average cover ratio obtained by an approximation algorithm is equal to the average optimal cover ratio
in the large-$N$ limit.
The main aim of this paper is to evaluate the typical performance threshold of $\rho_x$ below which an approximate
algorithm exhibits good typical performance with $L$ and $R$ fixed.
Evaluation of the approximation algorithms is accomplished by comparing their typical performance thresholds.

\begin{widetext}
\section{BP equations for Max-COV}\label{sec_bp}
As explained in this section, BP equations for max-COV are derived based on 
the statistical--mechanical model for the problem. 
We set particles on vertices in $V_1$ and $V_2$, which respectively occupy
vertex $i$ and $a$ if $x_i=1$ and $y_a=1$.
The hard-core lattice-gas model for max-COVs on bipartite graph $G$ is
then naively given as the following partition function.
\begin{equation}
\Xi_0(\mu;G)=\sum_{\bm{x}\in\{0,1\}^N}\sum_{\bm{y}\in\{0,1\}^M}\exp\left(\mu\sum_{a=1}^My_a\right)H\left(K-\sum_{i=1}^Nx_i\right)
\prod_{a=1}^{M}H\left(\sum_{i\in\partial a}x_i-y_a \right), \label{eq_s2_1}
\end{equation}
Therein, $H(x)=1\,(x\ge 0),\,0\,(x<0)$ is the Heaviside step function and
$\partial a=\{i\in V_1\mid (i,a)\in E\}$ stands for a set of neighbors of vertex $a\in V_2$. 
In the partition function, $\mu$ represents a chemical potential for particles on $V_2$.
One can construct BP equations for this model.
However, it is inconvenient for practical use because of the constraint $\sum_{i=1}^Nx_i\le K$.

Therefore, we use the following alternative partition function of 
\begin{equation}
\Xi(\mu;G)=\sum_{\bm{x}\in\{0,1\}^N}\sum_{\bm{y}\in\{0,1\}^M}\exp\left(\mu'\sum_{i=1}^Nx_i+\mu\sum_{a=1}^My_a\right)
\prod_{a=1}^{M}H\left(\sum_{i\in\partial a}x_i-y_a \right),  \label{eq_s2_2}
\end{equation}
where $\mu'=\mu'(\mu,\rho_x;G)$ is a chemical potential for particles on $V_1$.
The ratio of $\mu'$ with respect to $\mu$ is defined as
\begin{equation}
\kappa = -\frac{\mu'}{\mu}. \label{eq_s2_2a}
\end{equation}
This parameter is regarded as a Lagrange multiplier for the constraint on the number of selected vertices.
The appropriate value of $\kappa$ must satisfy the condition given by
\begin{equation}
\rho_x(\mu,\kappa)=\rho_x,\quad \rho_x(\mu,\kappa)\equiv \left\langle\frac{1}{N}\sum_{i=1}^N x_i \right\rangle_{\mu}, \label{eq_s2_3}
\end{equation}
where $\langle\dots\rangle_\mu$ is a grand-canonical average with a given $\mu$.
To consider ground states, we take the large-$\mu$ limit with parameter $\kappa$ fixed.
\end{widetext}

First, we construct BP equations for Eq.~(\ref{eq_s2_2}).
Using the Bethe--Peierls approximation, the single spin probability $P_i(x)$ which $x_i$ takes $x$ is
\begin{equation}
P_i(x)\simeq Z_i^{-1}e^{-\mu\kappa x}\prod_{a\in\partial i}P_{a\rightarrow i}(x),\label{eq_s2_4}, \\
\end{equation}
where $\partial i=\{a\in V_2\mid (i,a)\in E\}$ is a set of neighbors of vertex $i\in V_1$.
$Z_{\ast}$ is a normalization factor hereinafter.
$P_{a\rightarrow i}(x)$ represents the marginal probability of $\bm{x}_{\partial a\backslash i}$ and $y_a$ under the condition $x_i=x$.
Similarly, single spin probability $P_a(y)$ that $y_a$ takes $y$ reads
\begin{equation}
P_a(y)\simeq Z_a^{-1}e^{\mu y}\sum_{\bm{x}_{\partial a}}H\left(\sum_{i\in\partial a}x_i-y \right)\prod_{i\in\partial a}P_{i\rightarrow a}(x),\label{eq_s2_5}
\end{equation}
where $P_{i\rightarrow a}(x)$ is the probability of $x_i$ taking $x$ on
the cavity graph $G\backslash a$.
These probabilities satisfy the following recursive relations under the same approximation.
\begin{widetext}
\begin{align}
P_{i\rightarrow a}(x)&\simeq Z_{i\rightarrow a}^{-1}e^{-\mu\kappa x}\prod_{b\in\partial i\backslash a}P_{a\rightarrow i}(x),\label{eq_s2_6}\\
P_{a\rightarrow i}(x)&\simeq Z_{a\rightarrow i}^{-1}\sum_{y}e^{\mu y}\sum_{\bm{x}_{\partial a\backslash i}}H\left(x+\sum_{j\in\partial a\backslash i}x_j-y \right)
\prod_{j\in\partial a\backslash i}P_{j\rightarrow a}(x). \label{eq_s2_7}
\end{align}
\end{widetext}

To take the large-$\mu$ limit later, it is convenient to introduce cavity fields $\{h_{ia}\}$ and $\{\hat{h}_{ai}\}$
 defined respectively by $P_{i\rightarrow a}(x)\propto \exp(\mu h_{ia}x)$ and $P_{a\rightarrow i}(y)\propto \exp(\mu \hat{h}_{ai}y)$.
Then, BP equations for cavity fields are given as
\begin{align}
h_{ia}&= -\kappa+\sum_{b\in\partial i\backslash a}\hat{h}_{bi},\label{eq_s2_8}\\
\hat{h}_{ai}&= -\frac{1}{\mu}\ln\left[1-\frac{1}{1+e^{-\mu}}\prod_{j\in\partial a\backslash i} \frac{1}{1+e^{\mu h_{ja}}}\right]. \label{eq_s2_9}
\end{align}
By rescaling the single spin probability as $P_{i}(x)\propto \exp(\mu \xi_{i}x)$
 and $P_{a}(y)\propto \exp(\mu \hat{\xi}_{a}y)$ using local fields
 $\{\xi_i\}$ and $\{\hat{\xi}_a\}$, 
Eqs.~(\ref{eq_s2_4}) and (\ref{eq_s2_5}) then read
\begin{align}
\xi_{i}&= -\kappa+\sum_{a\in\partial i}\hat{h}_{ai},\label{eq_s2_8s}\\
\hat{\xi}_{a}&= -\frac{1}{\mu}\ln\left[1-\frac{1}{1+e^{-\mu}}\prod_{i\in\partial a} \frac{1}{1+e^{\mu h_{ia}}}\right]. \label{eq_s2_9s}
\end{align}
One can estimate the single spin probability
 by solving BP equations~(\ref{eq_s2_8}) and (\ref{eq_s2_9}) as the loopy
 belief propagation.  
\section{Replica-Symmetric Solution}\label{sec_rs}
In this section, typical performance of BP is studied using the RS cavity method based on the simplest RS ansatz.
The RS ansatz assumes that cavity fields $\{h\}$ and $\{\hat{h}\}$ are independent random variables 
respectively following probability distributions $P(h)$ and $\hat{P}(\hat{h})$.
For biregular random graphs, it is apparent that these distributions have no variance
because of the absence of fluctuation of degree in $V_1$ and $V_2$.
We therefore introduce the cavity fields $h$ and $\hat{h}$ on $(L,R)$-biregular random graphs.
Using BP equations~(\ref{eq_s2_8}) and (\ref{eq_s2_9}), they satisfy the following RS cavity equations
in the large-$N$ limit as 
\begin{align}
h&= -\kappa+(L-1)\hat{h},\label{eq_s3_8r}\\
\hat{h}&= -\frac{1}{\mu}\ln\left[1-\frac{1}{(1+e^{-\mu})(1+e^{\mu h})^{R-1}}\right]. \label{eq_s3_9r}
\end{align}
Similarly, the local fields on $(L,R)$-biregular random graphs are set
to $\xi$ and $\hat{\xi}$, which satisfy
\begin{align}
\xi&= -\kappa+L\hat{h},\label{eq_s3_8sr}\\
\hat{\xi}&= -\frac{1}{\mu}\ln\left[1-\frac{1}{(1+e^{-\mu})(1+e^{\mu h})^{R}}\right]. \label{eq_s3_9sr}
\end{align}
According to Eq.~(\ref{eq_s2_3}), for a given $\rho_x$, the local field $\xi$ is determined as 
\begin{equation}
\frac{e^{\mu\xi}}{1+e^{\mu\xi}}=\rho_x. \label{eq_s3_10}
\end{equation}
Then, using Eqs.~(\ref{eq_s3_8sr}) and (\ref{eq_s3_10}), the appropriate parameter $\kappa$ is represented as
\begin{equation}
\kappa=L\hat{h}-\frac{1}{\mu}\ln r, \label{eq_s3_11}
\end{equation}
where $r=\rho_x/(1-\rho_x)$.

By substituting Eqs.~(\ref{eq_s3_8r}) and~(\ref{eq_s3_11}) to Eq.~(\ref{eq_s3_9r})
 and by introducing $x=e^{-\mu\hat{h}}$, one obtains the self-consistent equation as 
\begin{equation}
\left(1+r x\right)^{R-1}(1-x)=\frac{1}{1+e^{-\mu}}.  \label{eq_s3_12}
\end{equation}
The order of this solution changes depending on the value of $\rho_x$.
The solution $x$ vanishes as $\mu$ becomes large if $\rho_x<1/R$ holds.
Using Taylor expansion of the left-hand side of Eq.~(\ref{eq_s3_12}), it is estimated as
\begin{equation}
x= \frac{e^{-\mu}}{1-(R-1)r}+O(e^{-2\mu}). \label{eq_s3_13}
\end{equation}
Using Eq.~(\ref{eq_s3_9sr}), an average density of particles on $V_2$ reads
\begin{align}
\lim_{M\rightarrow\infty}\overline{\left\langle\frac{1}{M}\sum_{i=1}^M y_a \right\rangle_{\mu}}
&= \frac{1-(1+rx)^{-R}}{1-(1+rx)^{-R}+e^{-\mu}}\label{eq_s3_13a}\\
&= R\rho_x+O(e^{-\mu})\label{eq_s3_13b}
\end{align}
Taking the large-$\mu$ limit, the average cover ratio for $\rho_x<1/R$ is 
obtained as 
\begin{equation}
 \rho_y^{\mathrm{RS}}= R\rho_x. 
  \label{eq_s3_14}
\end{equation}
This relation indicates that each vertex in $V_2$ is connected to at most one chosen vertex in $V_1$.
However, if $\rho_x$ is larger than $1/R$ ($R\ge 2$), then the
solution $x$ remains constant for sufficiently large $\mu$. 
In this case, a simple solution $\rho_y^{\mathrm{RS}}=1$ is obtained by Eq.~(\ref{eq_s3_13a}) 
because the average cover ratio in Eq.~(\ref{eq_s3_14}) touches to one at $\rho_x=1/R$.

Before closing this section, we discuss the stability of the RS solutions 
using spin-glass susceptibility defined as
\begin{equation}
\chi_{\rm SG}(\mu)=\frac{1}{N+M}\overline{\sum_{i\in V_1}\sum_{a\in V_2}
(\langle x_iy_a \rangle_{\mu}-\langle x_i \rangle_{\mu}\langle y_a
\rangle_{\mu})^2}. \label{eq_s3_chi1}
\end{equation}
Another representation of $\chi_{\rm SG}$ using cavity fields is known based on the fluctuation--dissipation theorem~\cite{Rivoire2003}.
As shown in~\cite{Mezard2007}, in the case of biregular random graphs, it reads
\begin{equation}
\chi_{\rm SG}(\mu)\simeq \sum_{d=0}^{\infty} \lambda^d,\quad
\lambda=\mathbb{E}\left[{\sum_{j\in\partial a\backslash i;\, b\in\partial j\backslash a}
\left(\frac{\partial \hat{h}_{ai}}{\partial \hat{h}_{bj}}\right)^2}\right],
\label{eq_s3_chi2}
\end{equation}
where $\mathbb{E}[\cdots]$ represents an average over vertices in $V_1$ and random graphs.
The cavity fields are mutually correlated in general.
 Actually, BP cannot converge any more~\cite{Zdeborova2009a} if the susceptibility diverges.
In this sense, the divergence of $\chi_{\rm SG}$ not only means the instability of the RS solution,
 but also poor typical performance of BP as an approximation algorithm in the static sense.
In the low-density region where $\rho_x<1/R$ holds, it is evaluated as
\begin{align}
\lambda &= (L-1)(R-1)\left[\frac{r(1-x)}{1+rx}\right]^2, \nonumber\\
&\rightarrow  (L-1)(R-1)r^2 \quad(\mu\rightarrow\infty).
\label{eq_s3_chi3}
\end{align}
The RS-RSB threshold in the large-$\mu$ limit therefore reads
\begin{equation}
\rho_x^{\mathrm{RS}}=\frac{1}{1+\sqrt{(L-1)(R-1)}}.\label{eq_s3_chi4}
\end{equation}
Otherwise, the other threshold $\rho_x^{\mathrm{RS}'}(>1/R)$ satisfies
\begin{equation}
(L-1)(R-1)\left[\frac{\rho_x^{\mathrm{RS}'}(1-x^\ast)}{(1-\rho_x^{\mathrm{RS}'})(1+x^\ast)}\right]^2=1, 
\label{eq_s3_chi5}
\end{equation}
where $x^\ast$ is a solution of
\begin{equation}
\left(1+ \frac{\rho_x^{\mathrm{RS}'}}{1-\rho_x^{\mathrm{RS}'}}x^\ast\right)^{R-1}(1-x^\ast)=1. \label{eq_s3_chi6}
\end{equation}

In Fig.~\ref{fig_rs}, RS-RSB thresholds $\rho_x^{\mathrm{RS}}$ and $\rho_x^{\mathrm{RS}'}$ on $(3R,R)$-regular random graphs are shown
as a function of $R$.
Except for the case $R=1$, the RS-RSB thresholds separate RS and RSB
regions. The RSB region remains for a finite $\rho_x$ while they converge to zero in the large-$R$ limit.

\begin{figure}
\includegraphics[width=0.95\linewidth]{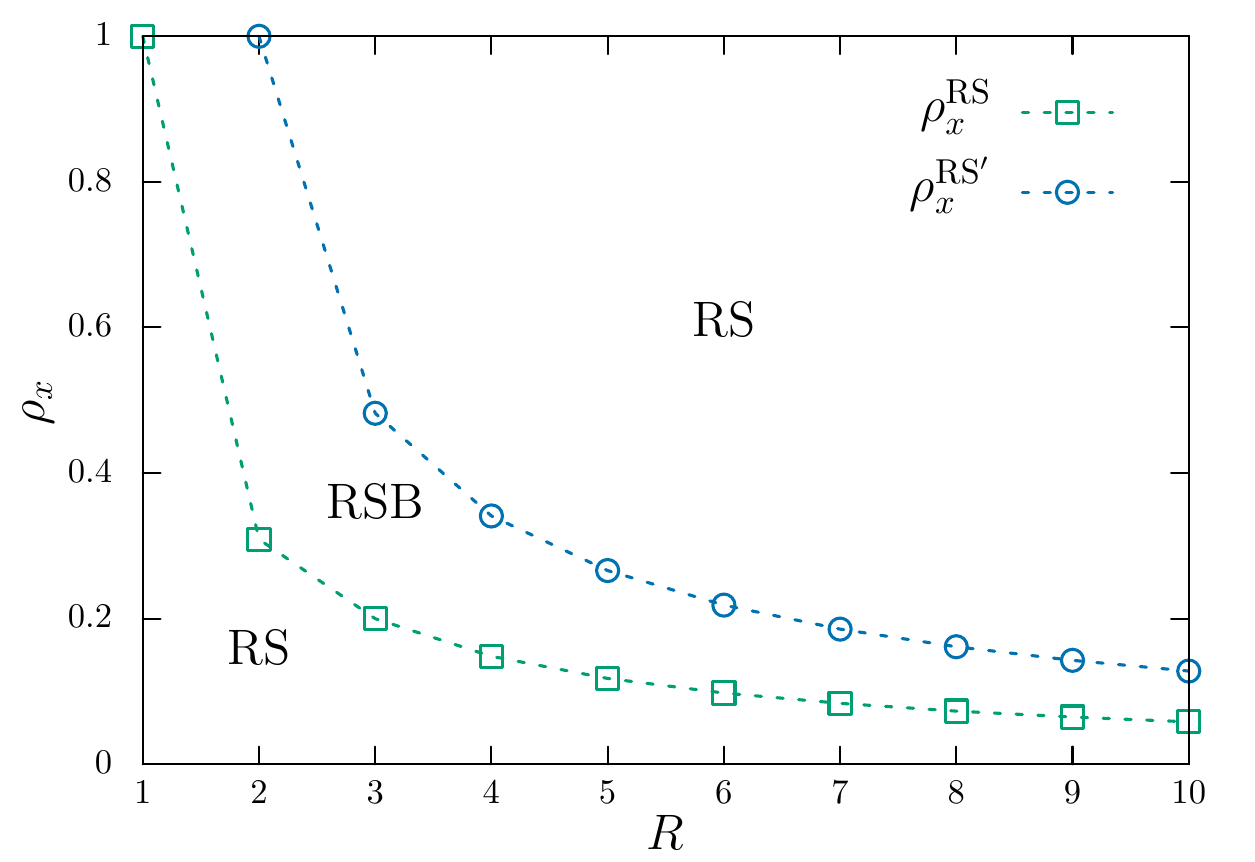}
\caption{RS-RSB thresholds $\rho_x^{\mathrm{RS}}$ (squares) and $\rho_x^{\mathrm{RS}'}$ (cycles)
 on $(3R,R)$-regular random graphs.}\label{fig_rs}
\end{figure}

\section{Typical Analysis of Greedy Algorithm}\label{sec_gre}
In this section, typical behavior of the greedy algorithm on $(L,R)$-biregular random graphs
 is investigated in the low-density region of $\rho_x<1/R$.
The corresponding RS solution indicates that one can choose vertices in $V_1$ without overlapping their neighbors.
It is therefore sufficient to analyze the fraction of vertices in $V_1$ with maximum degree
 during the deletion process.
Here we analyze a rate equation that is used frequently for analyses of similar greedy algorithms~\cite{Karp1981,Weigt2002}.
Let $V(T)$ be the expected number of vertices with maximum degree in
 $V_1$ at $T$-th step of the algorithm.
By the definition of the algorithm in sec.~\ref{sec_app}, we find that
\begin{equation}
NV(T+1)
=NV(T)-1-L(R-1)\frac{NV(T)}{N-RT}+O(N^{-1}), \label{eq_s4_gre1}
\end{equation}
where the assumption $L=O(1)$ is used.
By introducing $v(t)=V(T)/N$ and $t=T/N$, 
we obtain the following differential equation in the large graph limit. 
\begin{equation}
\frac{dv(t)}{dt}=-1-\frac{L(R-1)}{1-Rt}v(t).\label{eq_s4_gre2}
\end{equation}
Under the initial condition $v(0)=1$, the solution reads
\begin{equation}
v(t)=\frac{(L-1)(R-1)(1-Rt)^{\frac{L(R-1)}{R}}-(1-Rt)}{LR-L-R}.\label{eq_s4_gre4}
\end{equation}
Let $\rho_x^{\mathrm{g}}$ be a threshold below which vertices with the maximum degree are left at the end of the algorithm.
Considering that $t$ represents a fraction of chosen vertices in $V_1$, we find $v(\rho_x^{\mathrm{g}})=0$, that is,
\begin{equation}
\rho_x^{\mathrm{g}}=\frac{1}{R}\left\{1-[(L-1)(R-1)]^{-\frac{R}{LR-L-R}}\right\}.\label{eq_s4_gre5}
\end{equation}
If $\rho_x<\rho_x^{\mathrm{g}}$, then no chosen vertex in $V_1$ has overlapped neighbors in $V_2$
resulting in good typical performance of the algorithm, i.e., $\rho_y^{\mathrm{g}}=R\rho_x$.
However, the algorithm typically mistakes selections of vertices in $V_1$ and
underestimates the average cover ratio.
Results show that $\rho_x^{\mathrm{g}}\le \rho_x^{\mathrm{RS}}$
 where the equality holds for $R=1$, 
 indicating that BP is better than the greedy algorithm in terms of the typical performance threshold.

\section{Analysis of LP relaxation}\label{sec_lp}
An LP-relaxed value of the max-COV on any biregular graph is evaluated exactly using the LP duality~\cite{Rockafellar1972} as follows.

\begin{theorem}{}\label{thm_1}
An LP-relaxed value of the max-COV on any $(L,R)$-biregular graph is $LK$ if $K\le 1/R$ holds.
\end{theorem}
\begin{proof}
$x_i=K/N$ ($1\le i\le N$), $y_a=RK/N$ ($1\le a\le M$) is a feasible solution of LP-relaxed max-COV~(\ref{eq_s1_3}).
The value of the cost function is then $RKM/N=LK$.

The Lagrangian function of Eq.~(\ref{eq_s1_3}) is written as
\begin{align}
L(\bm{x},\bm{y},{p},\bm{q})=&\sum_{a=1}^My_a+p\left(K-\sum_{i=1}^Nx_i\right)\nonumber\\
&+\sum_{a=1}^Mq_a\left(\sum_{i\in\partial a}x_i-y_a\right),\label{eq_s5_lp1}
\end{align}
where $p\in\mathbb{R}$ and $\bm{q}\in\mathbb{R}^M$.
Using the weak duality theorem, one finds the following inequalities.
\begin{align}
\max_{\bm{x},\bm{y}}\min_{p,\bm{q}}L(\bm{x},\bm{y},{p},\bm{q})&\le \min_{p,\bm{q}}\max_{\bm{x},\bm{y}}L(\bm{x},\bm{y},{p},\bm{q})\nonumber\\
&\le \max_{\bm{x},\bm{y}}L(\bm{x},\bm{y},{p},\bm{q}).\label{eq_s5_lp2}
\end{align}
Consequently, $\max_{\bm{x},\bm{y}}L(\bm{x},\bm{y},{p},\bm{q})$ is an upper bound of the LP-relaxed value.
Because Eq.~(\ref{eq_s5_lp1}) is represented by
\begin{equation}
L(\bm{x},\bm{y},{p},\bm{q})= \sum_{i=1}^N\left(-p+\sum_{a\in\partial i}q_a\right)x_i+\sum_{a=1}^M(1-q_a)y_a+pK,\label{eq_s5_lp3}
\end{equation}
where a solution $x_i=y_a=1$ ($1\le i\le N,1\le a\le M$) realizes the maximum of the function with $p,\bm{q}$ satisfying
$-p+\sum_{a\in\partial i}q_a\ge 0$ and $1-q_a\ge 0 (1\le a\le M)$.
Assuming that $q_a=q$ holds for $1\le a\le M$, one finds that
\begin{equation}
\max_{\bm{x},\bm{y}}L\left(\bm{x},\bm{y},p,{q}\right)=-pN+LNq+M-Mq+pK.\label{eq_s5_lp4}
\end{equation}
On the right-hand side, setting $(p,q)=(M/N,K/N)$ results in $LK$, which indicates that the upper bound of the LP-relaxed max-COV on $(L,R)$-biregular graph is $LK$.
Because the approximate solution denoted above achieves this bound, the LP-relaxed value is equivalent to the bound.
\qed
\end{proof}

This theorem claims that $\rho_y=R\rho_x$ ($\rho_x\le 1/R$), suggesting that LP relaxation typically finds 
good approximate values in the RS regime.
The LP-relaxed value is equivalent to $\rho_y^{\mathrm{RS}}$ in the high-density regime where $\rho_x>1/R$.

\section{Numerical results}\label{sec_num}
This section presents description of some numerical results for validation of the theoretical analyses.
Here, we set $(L,R)=(9,3)$ as an example.
Then, the RS-RSB threshold and the typical performance threshold of the greedy algorithm are evaluated as
$\rho_x^{\mathrm{RS}}=0.2$ and $\rho_x^{\mathrm{g}}=(1-2^{-4/5})/3=0.1418\dots$.
Biregular random graphs are generated based on implementation of the configuration model~\cite{Newman2001}.
At least $400$ random graphs are used to take a random graph average.

\subsection{Average cover ratio}\label{sec_num_ave}
To examine the validity of theoretical analyses on the typical
 performance of approximation algorithms, the average cover ratio $\rho_y$ is evaluated
 using several methods.

 {
 We employ a BP-guided decimation (BPD) algorithm~\cite{Mezard2009} as a variant of loopy BP
 to obtain a feasible solution satisfying all the constraints.
 The algorithm fixes a value of a variable based on a solution of BP equations~(\ref{eq_s2_7}) until all variables are fixed to either zero or one.
 When a variable is fixed, BP equations are updated by applying the fixed value.
 If the number of variables in $V_1$ fixed to one reaches to $K$, the remainders of $\bm{x}$ are immediately fixed to zero to satisfy the constraint.
 An approximate value of the algorithm is thus a function of the parameter $\kappa$, which depends on inputs, i.e., $\mu$, $K$ and a graph.
 Here, the BP equations with $\mu=20$ are solved iteratively up to $150$ steps, which enables the algorithm to fix a variable practically
 even if the iterations cannot reach to the RS fixed point.
 The parameter $\kappa$ is tuned to maximize an approximate value while the detail will be reported elsewhere.
 In addition, the RS estimation of average cover ratio $\rho_y^{\mathrm{RS}}$ in the
 RS regime $\rho_x<\rho_x^{\mathrm{RS}}$ is used for comparison to BPD
 though its typical performance is possibly evaluated directly as in~\cite{Ricci-Tersenghi2009}.
 }
 
 The theoretical result presented in section~\ref{sec_lp} is used for LP relaxation, which is valid 
 for arbitrary biregular graphs.
Average approximate values of the greedy algorithm are  estimated numerically with $N=10^3$.

 We also evaluate the average optimal value using
the replica exchange Monte Carlo (EMC) method~\cite{Hukushima1996}. 
Results show that single-spin flip updates of the model~(\ref{eq_s2_1})
takes a long relaxation time for equilibration even using the exchange MC,
indicating the existence of deep valleys of the free energy. 
Consequently, it is necessary to accelerate equilibration for a system with sufficiently
small chemical potential $\mu$. 
To avoid such slow relaxation, we consider an alternative lattice--gas model on bipartite graph $G$ represented as
\begin{align}
\Xi_1&(\mu;G)\nonumber\\
&=\sum_{\bm{x}\in\{0,1\}^N}\exp\left(\mu\sum_{a=1}^M\theta\left(\sum_{i\in\partial a}x_i\right)\right)
\delta\left(K,\sum_{i=1}^Nx_i\right),
  \label{eq_s6_1}
\end{align}
where $\theta(x)$ takes one if $x> 0$ and zero otherwise, and $\delta(x,y)$ is Kronecker's delta.
In this model, variables $\bm{y}$ are eliminated because $y_a=\theta(\sum_{i\in\partial a}x_i)$ holds
 in the large-$\mu$ limit.
The ground states of the alternative model therefore correspond to the optimal solutions of the max-COV on the same graph.
Moreover, single-spin flip updates in the model are substantially equivalent to multi-spin updates
 in the original model~(\ref{eq_s2_1}), which makes the relaxation time to equilibrium states markedly short. 
Because optimal values are invariant by replacing inequality constraint $\sum_{i=1}^N x_i\le K$ to equality one,
the equality constraint is adopted and Kawasaki dynamics for a density-conserved system~\cite{Kawasaki1966a} is applied.

\begin{figure}
\includegraphics[width=0.95\linewidth]{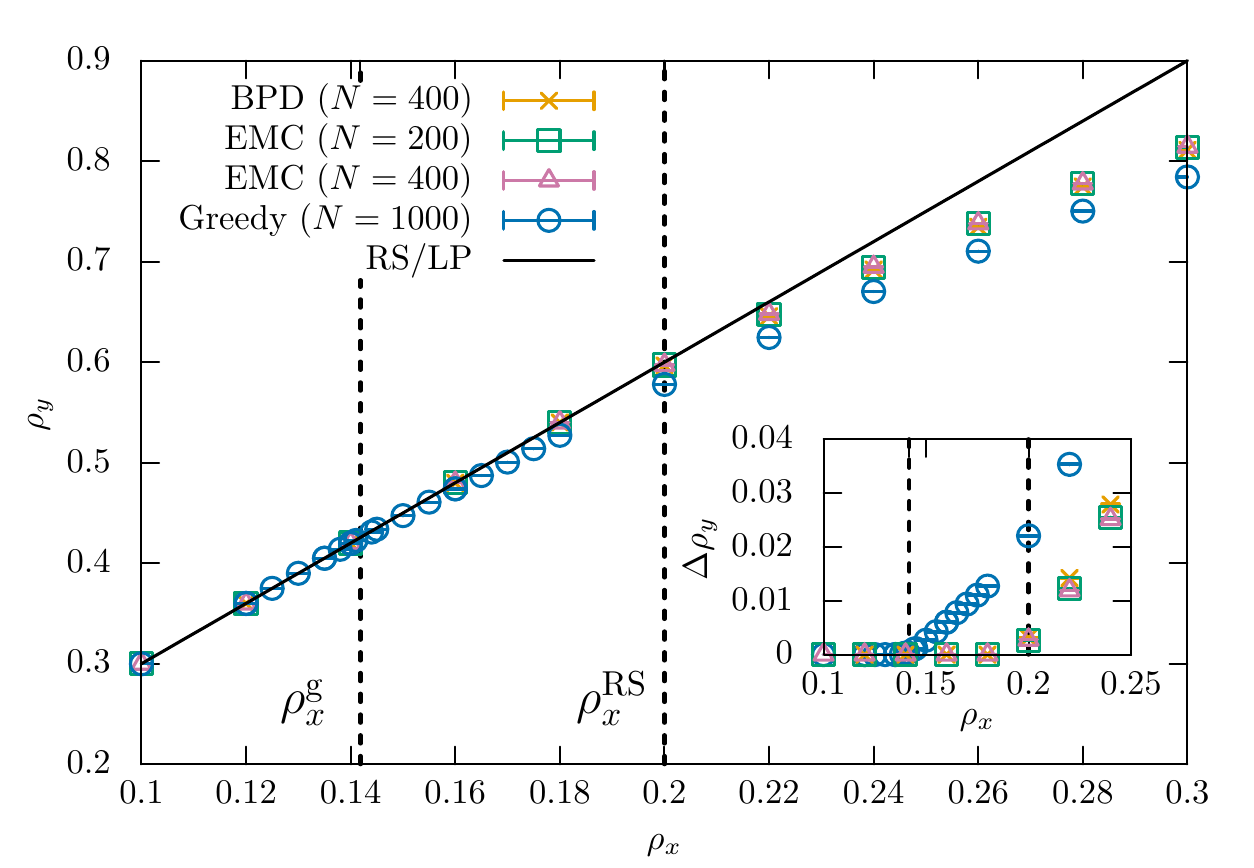}
\caption{Maximum cover ratio $\rho_y$ obtained by some approximation algorithms as a function of $\rho_x$. 
 Average optimal cover ratios are estimated for the cardinality $N=200$ (squares) and $400$ (triangles)
 by the replica exchange Monte Carlo (EMC) method, for $N=1000$ by the greedy algorithm (circles),
 and for $N=400$ by the {BP-guided decimation (BPD) algorithm} (cross marks).
 The solid line is the RS solution, equivalent to the LP relaxed value. 
Two vertical lines respectively represent the analytical performance threshold $\rho_x^{\mathrm{g}}$ of the greedy algorithm and
 the RS-RSB threshold $\rho_x^{\mathrm{RS}}$.
 (Inset) The difference $\Delta\rho_y$ between the 
 RS optimal cover ratio and
  numerical estimations by corresponding approximation algorithms as a function of $\rho_x$.}
 \label{fig_1}
\end{figure}

The results are presented in Fig.~\ref{fig_1}.
Numerical results are compatible to the RS estimation nearly up to the
RS-RSB threshold $\rho_x^{\mathrm{RS}}=0.2$ for the {BPD} and EMC,
and up to the typical performance threshold $\rho_x^{\mathrm{g}}$ predicted
analytically for the greedy algorithm.
As shown in the inset of Fig.~\ref{fig_1}, 
the average cover ratio by the greedy algorithm deviates from the 
average optimal ratio above $\rho_x^{\mathrm{g}}$, while other estimates
split above $\rho_x^{\mathrm{RS}}$, 
indicating that the Bethe--Peierls approximation and LP relaxation are no longer appropriate. 
In fact, as shown in Fig.~\ref{fig_lBP}, the loopy BP for $N=400$ fails its convergence above $\rho_x\simeq
0.195$, which may imply the existence of dynamical one-step RSB phase. 
{It is expected that BPD finds approximate solutions with good accuracy
 if a corresponding loopy BP converges to a fixed point.
Otherwise, decimation in BPD is performed based on a wrong estimation by the Bethe--Peierls approximation.
We emphasize, however, that BPD finds relatively good approximate solutions even in the RSB phase
as shown in Fig.~\ref{fig_1}.}
From these observations, we confirm that statistical--mechanical analyses 
{successfully predict typical performance thresholds by approximation algorithms.}

\begin{figure}
 \includegraphics[width=0.95\linewidth]{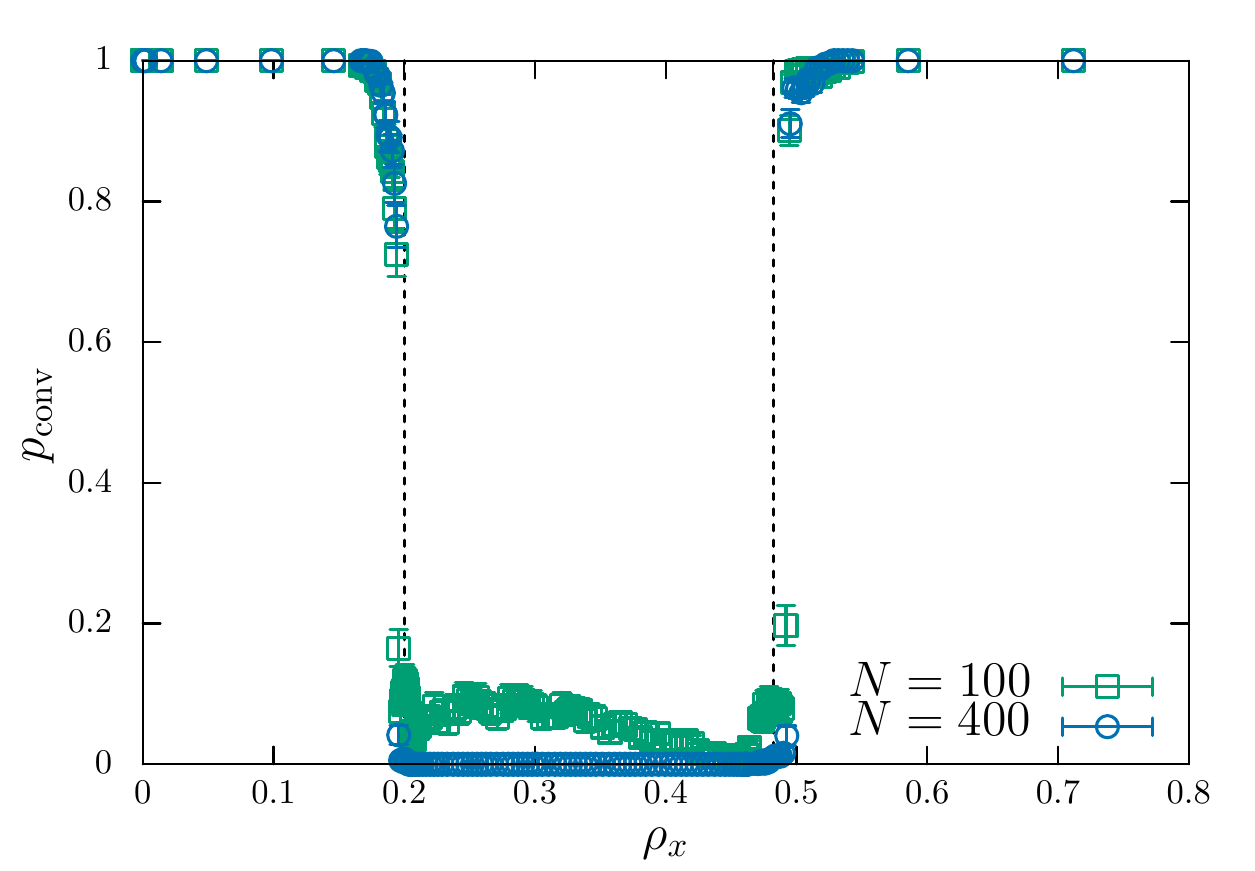}
 \caption{
 Probability of successfully convergent samples in the loopy BP for
 $(9,3)$-biregular random graphs as a
 function of $\rho_x$. The dotted vertical lines represents RS-RSB
 thresholds shown in Fig.~\ref{fig_rs}. 
 }
 \label{fig_lBP}
 \end{figure}

\subsection{Greedy algorithm and its variant}
Here we examine the greedy algorithm more closely.
To validate the correctness of our analysis, we specifically examine a fraction of
selected vertices without the maximum degree.
The value $r_{\mathrm{g}}$ is evaluated by
\begin{equation}
r_{\mathrm{g}}=1-\frac{\rho_x^{\mathrm{g}}}{\rho_x},\label{eq_g6_1}
\end{equation}
where $\rho_x^{\mathrm{g}}$ is given as Eq.~(\ref{eq_s4_gre5}).
The algorithm finds an optimal solution of the problem if $r_{\mathrm{g}}=0$.
As depicted in Fig.~\ref{fig_2}, the analytical estimations of $r_{\mathrm{g}}$ averaged over biregular random graphs
 agree well with the numerical results.
The fraction arises at $\rho_x^{\mathrm{g}}$ corresponding to the typical performance threshold.

Additionally, we modify the greedy algorithm to reduce the gap of its typical performance to that of BP.
As suggested above, the reason lies in the point that the greedy algorithm often chooses wrong vertices leading to 
shortage of vertices with maximum degree.
To avoid the situation, one must select a vertex to preserve as many
vertices to a maximum degree as possible,
meaning that the optimal selection requires an exhaustive search.
 We therefore propose a modified algorithm in consideration of
 the influence of the selection on other vertices in $V_1$,
  the simplest improved algorithm toward the optimal selection.
Let $\partial^2 i=\{j\in V_1\backslash i\mid j\in \partial a\,(\forall a\in\partial i)\}$
 be a set of the second neighbors of vertex $i\in V_1$.
We also define the subset of vertices with the maximum degree in $V_1$ by $W_1$.
The modified greedy algorithm is given as follows:
at each step, (i) choose vertex named $i\in W_1$ so that $|\partial^2i\cap W_1|$ is minimized,
(ii) delete vertices neighboring to vertex $i$ from $V_2$, and (iii) update $W_1$ and $V_1$ and return to (i) if $|V_1|>N-K$.
In Fig.~\ref{fig_2}, the fraction $r_{\mathrm{g}}$ of the modified
 algorithm is also shown with that of the original greedy algorithm.
 As shown in Fig.~\ref{fig_2}, the typical performance threshold of the proposed algorithm is
 improved, although it is still below the RS-RSB threshold.
This fact suggests that the typical performance of greedy algorithms depends of approximate choices and BP selects vertices better than 
those greedy selections.

\begin{figure}
\includegraphics[width=0.95\linewidth]{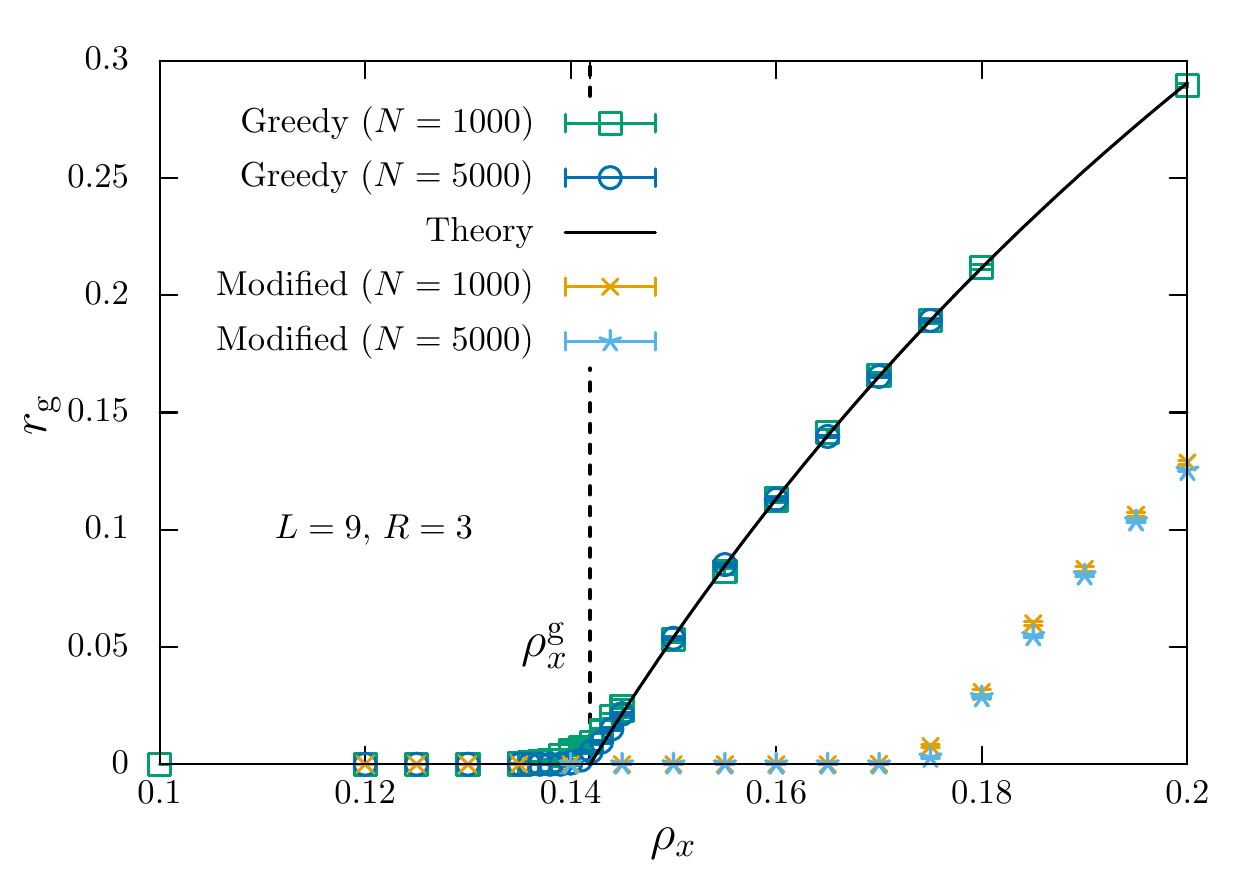}
\caption{Fraction $r_{\mathrm{g}}$ of selected vertices without maximum degree until the greedy algorithm stops as a function of $\rho_x$.
Symbols represent numerical results of the greedy algorithm with cardinality $N=1000$ (squares) and $5000$ (circles) 
and modified one with $N=1000$ (cross marks) and $5000$ (stars).
The solid line exhibits the analytical estimation of the fraction $r_g$.}\label{fig_2}
\end{figure}

\subsection{Randomized rounding of LP relaxation}
In Sec.~\ref{sec_num_ave}, we examine the typical performance of LP relaxation in terms of its approximate value.
In this subsection, randomized rounding, a practical means of constructing a feasible integer solution from LP relaxation, is applied to LP-relaxed solutions.
It is noteworthy that the typical performance of randomized rounding probably depends on the selection of an LP solver and its setting.
Here, we use IBM ILOG CPLEX with a default setting.
The approximate value of the rounded solution is compared to that of the RS solution in the RS phase, which is $LK$.
One considers that the rounding finds an optimal integer solution of the problem if two estimates mutually coincide.
We define a success ratio $p_r$ by the fraction of random graphs on which the rounding finds an optimal solution.
Fig.~\ref{fig_3} presents the average approximation ratio $\rho_y^{r}$ and the success ratio $p_r$ as a function of $\rho_x$, which strongly suggests that the randomized rounding exhibits a phase transition in terms of its typical performance.
The threshold $\rho_x^{\mathrm{r}}$ is less than $0.1$, which is much 
lower than $\rho_x^{\mathrm{g}}$,
whereas the LP-relaxed value is regarded as nearly optimal in the RS region $\rho_x<\rho_x^{\mathrm{RS}}=0.2$.
We therefore conclude that the typical performance of the randomized rounding
of LP-relaxed solution is inferior to that of the greedy algorithm.

\begin{figure}
\begin{center}
\includegraphics[width=0.95\linewidth]{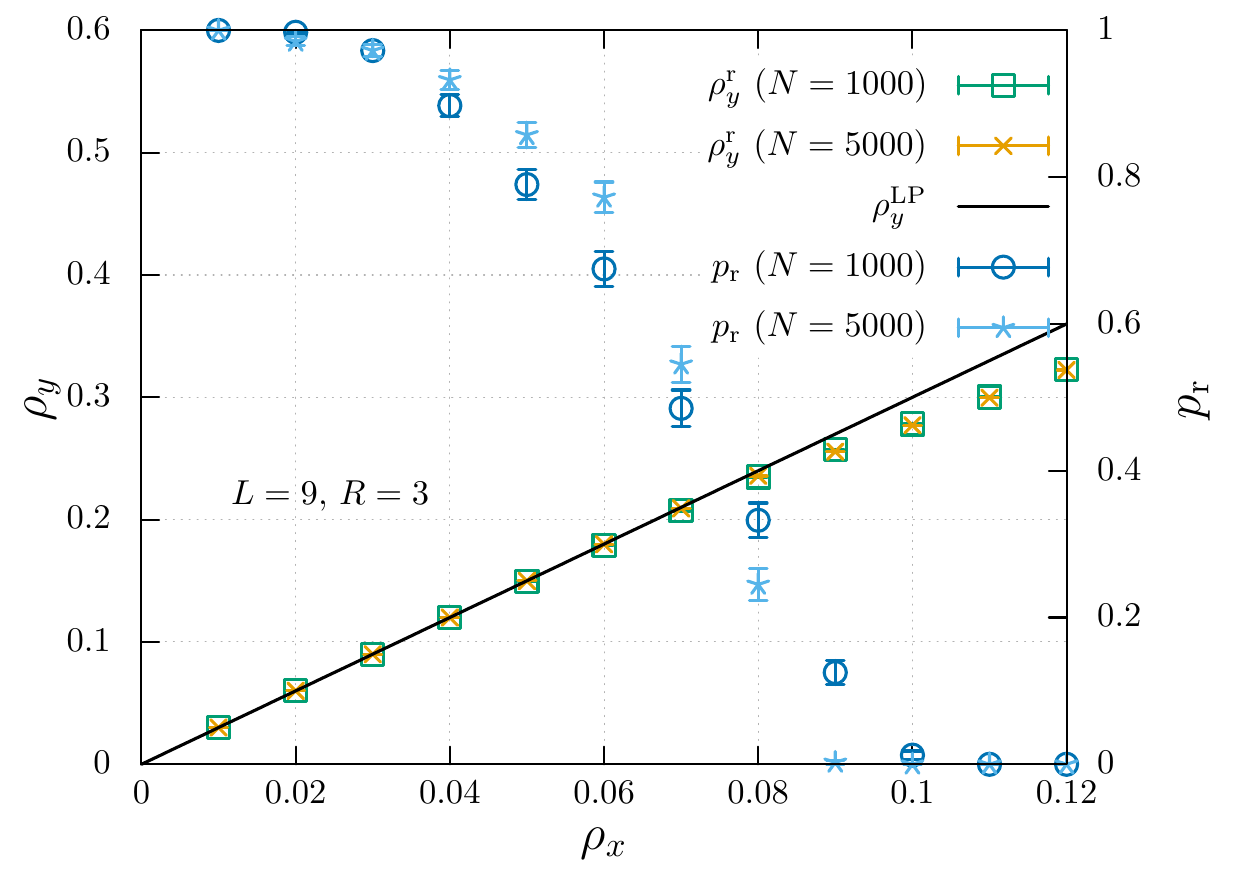}
\caption{Typical performance of randomized rounding of LP relaxation as a function of $\rho_x$ (left vertical axis). 
Symbols represent approximation values $\rho_y^{\mathrm{r}}$ of rounded LP-relaxed solutions with cardinality $N=1000$ (squares) and $5000$ (cross marks).
The solid line represents the average cover ratio equivalent to that of LP relaxation in the region shown.
Squares and stars respectively represent the success ratio $p_r$ (right vertical axis) of randomized rounding with $N=1000$ and $5000$.
}\label{fig_3}
\end{center}
\end{figure}

\section{Summary and Discussion}
As described in this paper, we investigate the typical performance of approximation algorithms called belief propagation, greedy algorithm, and
linear-programming relaxation for maximum coverage problem on sparse biregular random graphs.
The typical performance of BP is studied by application of the RS cavity method to a correspondent hard-core lattice--gas model.
Results show that, in the large-$\mu$ limit, there exist two distinct RS-RSB thresholds regarded as typical performance thresholds of BP.
In addition, the greedy algorithm performance and LP relaxation were studied especially in the low-density region.
Results show that the typical performance threshold of the greedy algorithm is lower than that of BP
and that LP-relaxed values are always equivalent to the RS solutions leading to the threshold equivalent to that of BP.
Those analytical results were validated by executing some numerical simulations.
Results of additional numerical studies suggest that BP typically works better than the modified greedy algorithm and
 that randomized rounding of LP-relaxed solutions has a lower threshold than the greedy algorithm.

To assess the typical performance of BP as an approximation algorithm, we concentrated on statistical--mechanical analysis of max-COV
up to the RS level.
Further analyses based on the one step RSB will be necessary to reveal statistical--mechanical properties of the problem and
typical performance analysis of another algorithm called survey propagation.
Another possible avenue of future work is the extension of our analyses to other random bipartite graphs.
As with min-VCs, it is an attractive question whether the magnitude relation of typical performance thresholds changes depending on
the random graph ensembles.

Our results provide not only respective typical performance of approximation algorithms but also their suggestive mutual relations.
Specifically examining LP relaxation, it is worth emphasizing that the LP relaxation finds good approximate values compared to optimal values but 
the typical performance of its randomized rounding has the smallest threshold among approximation algorithms studied here.
To fill the gap of thresholds, it is important to examine modifications of LP relaxation such as the cutting-plane approach~\cite{Schrijver1998}.
As for the greedy algorithms and their modification, numerical results suggest that evaluation of the influence of their deletion process
 affects the marked improvement of the typical performance threshold.
The fact that BP is better than the greedy algorithm and its modification
indicates that BP incorporates the influence more efficiently.
These suggestions are expected to be of great help to understand properties and relations of approximation algorithms in terms of typical performance.
In addition, our analyses of the greedy algorithm and randomized
rounding of LP relaxation illustrate that
typical case and worst case evaluations capture different notions of approximate performance in optimization problems. 
This fact indicates the importance of the typical-case analysis of approximation algorithms.
We hope that the arguments and results presented herein stimulate further studies and that the typical performance analyses of approximation algorithms
 will attract the interest of researchers in many diverse fields.

\begin{acknowledgements}
ST warmly thanks T. Takaguchi for stimulating this study.
The use of IBM ILOG CPLEX has been supported by the IBM Academic Initiative.
This research was supported by JSPS KAKENHI Grant Nos.~25120010 (KH),~16K16011 (TM), and~15J09001 (ST).
\end{acknowledgements}

\bibliography{MaxCov_typ_2}

\end{document}